\documentclass[letterpaper, 10 pt, conference]{ieeeconf} 
\IEEEoverridecommandlockouts 
\usepackage{times} 
\usepackage{enumerate}
\usepackage{mathtools}
\usepackage{amsmath}
\usepackage{amssymb}
 
\usepackage{amsthm}
\usepackage{graphicx}
\usepackage{xcolor}
\usepackage{url}
\usepackage{algorithm2e}
\usepackage{cite}

\usepackage{etoolbox}

\usepackage{abbreviations}

\usepackage{siunitx}
\usepackage{booktabs, makecell, multirow}
\usepackage{aligned-overset}
\usepackage{tikz}
\usetikzlibrary{shapes,arrows,positioning,fit,backgrounds,calc,plotmarks,decorations.pathreplacing,intersections,pgfplots.groupplots}
\usepackage{pgfplots}

\newcommand\copyrighttext{%
  \footnotesize 
  \textcopyright 2023 IEEE. Personal use of this material is permitted.
  Permission from IEEE must be obtained for all other uses, in any current or future
  media, including reprinting/republishing this material for advertising or promotional
  purposes, creating new collective works, for resale or redistribution to servers or
  lists, or reuse of any copyrighted component of this work in other works.
}
\newcommand\copyrightnotice{%
    \begin{tikzpicture}[remember picture,overlay]
        \node[anchor=south,yshift=10pt] at (current page.south) {\fbox{\parbox{\dimexpr\textwidth-\fboxsep-\fboxrule\relax}{\copyrighttext}}};
    \end{tikzpicture}%
}

\allowdisplaybreaks

\newtheorem{theorem}{Theorem}

\newtheorem{proposition}[theorem]{Proposition}

\newtheorem{definition}[theorem]{Definition}
\newtheorem{remark}[theorem]{Remark}
\newtheorem{assumption}[theorem]{Assumption}
\newtheorem{example}[theorem]{Example}

\title{\LARGE \bf
    Control of bilinear systems using gain-scheduling%
    : Stability and performance guarantees
}

\author{Robin Str\"asser, Julian Berberich, Frank Allg\"ower%
\thanks{F.\ Allgöwer is thankful that this work was funded by the Deutsche Forschungsgemeinschaft (DFG, German Research Foundation) under Germany's Excellence Strategy -- EXC 2075 -- 390740016 and within grant AL 316/15-1 -- 468094890. 
R.\ Strässer thanks the Graduate Academy of the SC SimTech for its support.}%
\thanks{R.\ Str\"asser, J.\ Berberich, and F.\ Allg\"ower are with the Institute for Systems Theory and Automatic Control, University of Stuttgart, 70550 Stuttgart, Germany
		(email:\{straesser,berberich,allgower\}@ist.uni-stuttgart.de)}%
}

\begin{document}
\maketitle
\copyrightnotice
\thispagestyle{empty}
\pagestyle{empty}
 

\begin{abstract}%
    In this paper, we present a state-feedback controller design method for bilinear systems.
    To this end, we write the bilinear system as a linear fractional representation by interpreting the state in the bilinearity as a structured uncertainty. 
    Based on that, we derive convex conditions in terms of linear matrix inequalities for the controller design, which are efficiently solvable by semidefinite programming. 
    Further, we prove asymptotic stability and quadratic performance of the resulting closed-loop system locally in a predefined region.
    The proposed design uses gain-scheduling techniques and results in a state feedback with rational dependence on the state, which can substantially reduce conservatism and improve performance in comparison to a simpler, linear state feedback.
    Moreover, the design method is easily adaptable to various scenarios due to its modular formulation in the robust control framework.
    Finally, we apply the developed approaches to numerical examples and illustrate the benefits of the approach. 
\end{abstract}

\section{INTRODUCTION}\label{sec:introduction}
Nonlinear control theory is a branch of control engineering that deals with the analysis and design of control systems that exhibit nonlinear behavior. In contrast to linear control systems, nonlinear control systems have complex and sometimes unpredictable dynamics that make their analysis and control challenging~\cite{khalil:2002}. An important class of nonlinear systems is the class of bilinear systems, which finds many practical applications in, e.g., engineering, power systems, biology, economics, and ecology (cf.~\cite{mohler:1973,bruni:dipillo:koch:1974,spinu:athanasopulos:lazar:bitsoris:2012,espana:landau:1978}). 
Moreover, bilinear models have recently gained an increasing interest since linearization techniques based on a higher-dimensional lifting, such as, e.g., Carleman linearization~\cite{carleman:1932} or Koopman operator theory~\cite{koopman:1931}, transform general nonlinear systems into possibly infinite-dimensional bilinear control systems. 

As a brief overview, we outline some methods available in the literature for controller design for bilinear systems. 
Existing approaches use, e.g., Lyapunov's second method~\cite{pedrycz:1980}, bang-bang control with linear switching policy~\cite{longchamp:1980}, or quadratic state feedback~\cite{gutman:1981,gutman:1980}. While these rely on checking multivariate polynomial equations, a more practical Lyapunov-based state-feedback controller was proposed in~\cite{derese:noldus:1980}. 
Further approaches include constant feedback~\cite{luesink:nijmeijer:1989} and optimization-based nonlinear state feedback~\cite{benallou:mellichamp:seborg:1988}, or rely on a version of the Kalman-Yakubovich-Popov lemma for passive bilinear systems~\cite{lin:byrnes:1994}. 
Linear matrix inequality (LMI)-based local stabilization techniques are presented in~\cite{amato:cosentino:fiorillo:merola:2009} for a polytopic region and in~\cite{khlebnikov:2018} for an ellipsoidal region, where the latter uses Petersen's lemma~\cite{petersen:1987,khlebnikov:shcherbakov:2008} to interpret parts of the bilinearity as an uncertainty. 
A related approach is taken in~\cite{strasser:berberich:allgower:2023}, which derives regional closed-loop stability guarantees for bilinear system under linear state feedback by reformulating the system as a linear fractional representation (LFR).
Alternatively, bilinear systems can also be viewed as (quasi-)linear parameter-varying systems when considering the state as scheduling variable, which allows to derive convex controller design conditions (see, e.g.,~\cite{huang:jadbabaie:1999}).
A different approach for regional stabilization of nonlinear systems is, e.g., presented in~\cite{coutinho:desouza:dasilva:caldeira:prieur:2019} for input-delayed uncertain polynomial systems.

The contribution of the present paper is to derive a novel controller design method for bilinear systems using gain-scheduling techniques. 
Our approach leads to an LMI feasibility problem which results in a state-feedback controller with rational dependence on the state.
Then, we prove local stability of the bilinear system in closed loop. 
The design and its theoretical analysis are based on an LFR, where, inspired by~\cite{khlebnikov:2018}, the state in the bilinearity is interpreted as uncertainty.
This idea is similar to~\cite{strasser:berberich:allgower:2023} with the main difference that the proposed approach allows for a significantly more flexible controller parametrization and for considering multi-dimensional inputs.
More precisely, we use gain-scheduling techniques (compare~\cite{scherer:2001,veenman:scherer:2014,lawrence:rugh:1993,leith:leithead:2000}), which allows to design controllers with rational dependence on the state.
Additionally, we derive more general multipliers to describe the bilinear terms that are beneficial for multiple inputs and reduce conservatism of the controller design. 
Further, we extend these results to design controllers with closed-loop performance guarantees, e.g., on the $\cL_2$-gain, and demonstrate the improvements of the proposed design approaches.

The paper is organized as follows. 
In Section~\ref{sec:linear-representation-bilinear}, we first represent bilinear systems locally as an LFR by interpreting the state in the bilinearity as a structured uncertainty. 
Then, Section~\ref{sec:controller-design-bilinear} contains a state-feedback controller design procedure guaranteeing closed-loop stability for bilinear systems based on gain-scheduling.
In Section~\ref{sec:controller-performance}, we extend the design procedure and derive guarantees for asymptotic stability and quadratic performance of the closed loop. 
The developed controller is applied to three numerical examples in Section~\ref{sec:simulation-examples}. 
Finally, we conclude the paper in Section~\ref{sec:conclusion}.

\paragraph*{Notation}
We write $I_p$ for the $p\times p$ identity matrix and $0_{p\times q}$ for the $p\times q$ zero matrix, where we omit the index if the dimension is clear from the context. 
If $A$ is symmetric, then we write $A\succ 0$ or $A\succeq 0$ if $A$ is positive definite or positive semidefinite, respectively. 
Negative (semi)definiteness is defined analogously. 
Matrix blocks which can be inferred from symmetry are denoted by $\star$ and we abbreviate $B^\top A B$ by writing $[\star]^\top AB$. 
Finally, $\kron$ denotes the Kronecker product.

\section{LINEAR FRACTIONAL REPRESENTATION OF BILINEAR SYSTEMS}\label{sec:linear-representation-bilinear}
In this section, we introduce the problem setting of this paper (Section~\ref{sec:problem-setting}) and derive an LFR which represents the bilinear system (Section~\ref{sec:LFR-nominal}).

\subsection{Problem setting}\label{sec:problem-setting}
This paper considers discrete-time bilinear systems of the form
\begin{equation}\label{eq:dynamics-bilinear}
    z_+ 
    = A z + B_0 u + \sum_{j=1}^m u_j B_{j} z
    = Az + B_0 u + \tB (u\kron z),
\end{equation}
where $z,z_+\in\bbR^N$, $u\in\bbR^{m}$, $A\in\bbR^{N\times N}$, $B_0\in\bbR^{N\times m}$, $B_j\in\bbR^{N\times N}$, $j=1,...,m$, and 
$
    \tB = \begin{bmatrix}
        B_1 & \cdots & B_m
    \end{bmatrix}
$. 
Our goal is to design a state-feedback controller $u=k(z)$ such that the closed-loop system $z_+=Az + B_0k(z) + \tB(k(z)\kron z)$ satisfies local stability (Section~\ref{sec:controller-design-bilinear}) and performance (Section~\ref{sec:controller-performance}) requirements. More specifically, we want to achieve local closed-loop guarantees for all initial conditions $z_0$ in an ellipsoidal region $\cZ_\mathrm{RoA}$ defined later.
To this end, we define the set $\cZ = \{z\in\bbR^N \mid \eqref{eq:condition-z-in-Z}~\text{holds}\}$ based on the quadratic inequality 
\begin{equation}\label{eq:condition-z-in-Z}
    \begin{bmatrix}
        z \\ 1
    \end{bmatrix}^\top 
    \begin{bmatrix}
        Q_z & S_z \\ S_z^\top & R_z
    \end{bmatrix}
    \begin{bmatrix}
        z \\ 1
    \end{bmatrix}
    \geq 0,
\end{equation}
where $Q_z \prec 0$, and $R_z \succ 0$. 
The description of $\cZ$ includes, e.g., a region described by $z^\top z\leq c$ with $c>0$ when choosing $Q_z=-I$, $S_z=0$, and $R_z=c$.
Further, we assume that the inverse 
\begin{equation*}
    \begin{bmatrix}
        \tQ_z & \tS_z \\
        \tS_z^\top & \tR_z
    \end{bmatrix}
    \coloneqq 
    \begin{bmatrix}
        Q_z & S_z \\
        S_z^\top & R_z
    \end{bmatrix}^{-1}    
\end{equation*}
exists.
The later derived theoretical analysis will rely on the state $z$ being within $\cZ$ for all times. 
This will be ensured via a suitable Lyapunov function sublevel set $\cZ_\mathrm{RoA}\subseteq \cZ$.

\begin{remark}
    Although we consider only discrete-time bilinear systems in this paper, we conjecture that the proposed controller design and its theoretical guarantees can be translated to continuous-time bilinear systems.
\end{remark}

\subsection{Linear fractional representation of bilinear systems}\label{sec:LFR-nominal}
In the following, we reformulate the bilinear system~\eqref{eq:dynamics-bilinear} as an LFR.
In particular,~\eqref{eq:dynamics-bilinear} is equivalent to
\begin{subequations}\label{eq:LFR-bilinear-open-loop}
    \begin{align}
        \begin{bmatrix}
            z_+ \\ u
        \end{bmatrix}
        &= \begin{bmatrix}
            A & B_0 & \tB \\
            0 & I & 0
        \end{bmatrix}
        \begin{bmatrix}
            z \\ u \\ w
        \end{bmatrix},\\
        w &= (I_m\kron z) u \label{eq:LFR-bilinear-nominal-general-uncertainty}
    \end{align}
\end{subequations}
with $z\in\cZ$. 
An LFR as in~\eqref{eq:LFR-bilinear-open-loop} is a common representation of uncertain systems~\cite{zhou:doyle:glover:1996}. Here, the state $z$ is interpreted as an uncertainty and, thus, we reduce the bilinear control problem to a linear control problem with nonlinear state-dependent uncertainty, i.e., the state $z$.
More precisely, the LFR is exposed to the uncertainty $(I_m\kron z)$ for which we need a suitable uncertainty characterization.
In order to formulate the characterization in a tractable way, we define the set 
\begin{equation}\label{eq:Delta-representation-Delta-structured}
    \mathbf{\Delta} 
    \coloneqq 
    \left\{
        \Delta \in \bbR^{mN\times m}
    \middle|
        \begin{bmatrix}
            \Delta \\ I
        \end{bmatrix}^\top 
        \Pi_{\Delta}
        \begin{bmatrix}
            \Delta \\ I
        \end{bmatrix}
        \succeq 0
        \;\forall\,
        \Pi_{\Delta} \in \mathbf{\Pi}_{\Delta}
    \right\}
\end{equation}
for some multiplier class $\mathbf{\Pi}_{\Delta}$ defined in the following as a convex cone of symmetric matrices.
Incorporating the knowledge that the uncertainty in~\eqref{eq:LFR-bilinear-open-loop} is of the form $I_m\kron z$ with $z\in\cZ$, we choose $\mathbf{\Pi}_\Delta$ such that $(I_m\kron z)\in\mathbf{\Delta}$.
To this end, we propose to choose the multiplier class via the LMI representation%
\begin{equation}\label{eq:Delta-representation-Pi-structured}\small
    \mathbf{\Pi}_{\Delta} \coloneqq 
    \left\{
        \Pi_{\Delta} 
    \middle|
        \Pi_{\Delta} = \begin{bmatrix}
            \Lambda \kron Q_z & \Lambda \kron S_z \\
            \Lambda \kron S_z^\top & \Lambda \kron R_z
        \end{bmatrix}, 0\preceq\Lambda\in\bbR^{m\times m}
    \right\}.
\end{equation}
Inspired by~\cite[Prop. 1]{berberich:scherer:allgower:2022}, we make the following observation.
\begin{proposition}\label{prop:multiplier-tightness}
    Let $\mathbf{\Pi}_{\Delta}$ be as in~\eqref{eq:Delta-representation-Pi-structured}. 
    Then, $\Delta\in\mathbf{\Delta}$ if and only if $\Delta=I_m \kron z$ with $z\in\cZ$.
\end{proposition}
\begin{proof}
    \emph{"If":~}
    Suppose $z\in\cZ$, i.e.,~\eqref{eq:condition-z-in-Z} holds, and let $0\preceq \Lambda \in \bbR^{m\times m}$ be arbitrary. Then, $(I_m\kron z)\in\bbR^{mN\times m}$ satisfies  
    \begin{multline*}
        \begin{bmatrix}
            (I_m\kron z) \\ I
        \end{bmatrix}^\top 
        \begin{bmatrix}
            \Lambda \kron Q_z & \Lambda \kron S_z \\
            \Lambda \kron S_z^\top & \Lambda \kron R_z
        \end{bmatrix}
        \begin{bmatrix}
            (I_m\kron z) \\ I
        \end{bmatrix}
        \\
        = \Lambda \kron \left(
            \begin{bmatrix}
                z \\ 1
            \end{bmatrix}^\top
            \begin{bmatrix}
                Q_z & S_z \\ S_z^\top & R_z
            \end{bmatrix}
            \begin{bmatrix}
                z \\ 1
            \end{bmatrix}
        \right)
        \succeq 0
    \end{multline*}
    due to $\Lambda\succeq 0$ and~\eqref{eq:condition-z-in-Z}, where we use the spectral property of the Kronecker product that $A\kron B$ has eigenvalues $\lambda_i\mu_j$, $i=1,...,m$, $j=1,...,n$ if $A\in\bbR^{m\times m}$ has eigenvalues $\{\lambda_i\}_{i=1}^m$ and $B\in\bbR^{n\times n}$ has eigenvalues $\{\mu_j\}_{j=1}^n$. Since $\Lambda\succeq 0$ was arbitrary, this shows $(I_m\kron z)\in\mathbf{\Delta}$.

    \emph{"Only if":~}
    See Appendix~\ref{sec:appendix-proof-prop-multipliers}. 
\end{proof}

This proposition shows that we can restrict ourselves to unstructured uncertainties $\Delta\in\mathbf{\Delta}$ in the LFR~\eqref{eq:LFR-bilinear-open-loop} using the multiplier class $\mathbf{\Pi}_{\Delta}$ since the knowledge about the structure of $(I_m\kron z)\in\mathbf{\Delta}$ is already incorporated in $\mathbf{\Pi}_{\Delta}$. 
In particular, any $\Delta\in\mathbf{\Delta}$ is necessarily of the required form $\Delta=I_m\kron z$ with $z\in\cZ$. 
Hence, $\mathbf{\Delta}$ with multiplier class $\mathbf{\Pi}_{\Delta}$ exploits the structure of the uncertainty without additional conservatism enabling a controller design for multi-dimensional inputs. 
For scalar inputs, i.e., $m=1$, the uncertainty description $\mathbf{\Delta}$ reduces to $\cZ$, i.e., $\mathbf{\Delta}=\cZ$. 

\section{CONTROLLER DESIGN FOR BILINEAR SYSTEMS}\label{sec:controller-design-bilinear}
Next, we exploit the LFR derived in the last section to design controllers for bilinear systems. To this end, we first establish conditions for closed-loop stability using linear state feedback in Section~\ref{sec:controller-state-feedback}. In Section~\ref{sec:controller-gain-scheduling}, we use arguments from gain-scheduling to derive an improved controller design approach with higher flexibility and reduced conservatism. 

\subsection{State-feedback control law}\label{sec:controller-state-feedback}
We use the LFR representation~\eqref{eq:LFR-bilinear-open-loop} to design a stabilizing state-feedback controller $u=Kz$, $K\in\bbR^{m\times N}$. Plugging the control law into~\eqref{eq:LFR-bilinear-open-loop}, we obtain the closed-loop LFR 
\begin{subequations}\label{eq:LFR-bilinear-nominal}
    \begin{align}
        \begin{bmatrix}
            z_+ \\ u
        \end{bmatrix}
        &= \begin{bmatrix}
            A+B_0 K & \tB \\
            K & 0
        \end{bmatrix}
        \begin{bmatrix}
            z \\ w
        \end{bmatrix},\\
        w &= (I_m\kron z) u,
    \end{align}
\end{subequations}
where $(I_m\kron z) \in \mathbf{\Delta}$.
The following theorem establishes closed-loop stability of the bilinear system~\eqref{eq:dynamics-bilinear} exploiting the structure of the uncertainty set $\mathbf{\Delta}$.

\begin{theorem}\label{thm:stability-condition-LFR-nominal}
    If there exist a symmetric $N\times N$ matrix ${P}=P^\top\succ 0$, a matrix ${L}\in\bbR^{m\times N}$, a symmetric $m\times m$ matrix $\tilde{\Lambda}=\tilde{\Lambda}\succ 0$, and a scalar $\nu> 0$ such that $\cQ \succ 0$, where
    \begin{equation*}
        \cQ =
        \begin{bmatrix}
            {P}
            & -\tB (\tilde{\Lambda}\kron\tS_z)
            & A {P} + B_0 L
            & \tB(\tilde{\Lambda}\kron \tQ_z)
            \\
            \star
            & \tilde{\Lambda}\kron\tR_z
            & L
            & 0
            \\
            \star 
            & \star
            & {P} 
            & 0
            \\
            \star 
            & \star
            & \star
            & -\tilde{\Lambda}\kron\tQ_z
        \end{bmatrix},
    \end{equation*}
    and 
    \begin{equation}\label{eq:stability-condition-LFR-nominal-invariance}
        \begin{bmatrix}
            \nu \tQ_z + P & - \nu \tS_z \\ -\nu \tS_z^\top & \nu \tR_z - 1
        \end{bmatrix}
        \preceq 0,
    \end{equation}
    then $\cZ_\mathrm{RoA} = \{z\in\bbR^N \mid z^\top P^{-1} z \leq 1\}\subseteq \cZ$ and the controller $u(z) = L {P}^{-1} z$ locally asymptotically stabilizes system~\eqref{eq:dynamics-bilinear} for all initial conditions $z_0\in\cZ_\mathrm{RoA}$.
\end{theorem}
\begin{proof}
    Since we later prove closed-loop stability of the bilinear system in a more general framework that includes a linear state-feedback design as a special case, we omit the proof at this point.
\end{proof}
The controller design in Theorem~\ref{thm:stability-condition-LFR-nominal} achieves local asymptotic stability with guaranteed region of attraction $\cZ_\mathrm{RoA}$. 
A similar condition for closed-loop stability of bilinear systems with scalar input signals ($m=1$) is given in~\cite[Thm.~1]{khlebnikov:2018} based on Petersen's lemma and in ~\cite[Thm. 4]{strasser:berberich:allgower:2023} using the LFR framework. 
Thus, Theorem~\ref{thm:stability-condition-LFR-nominal} extends this result to general input signals ($m\geq 1$) based on the LFR~\eqref{eq:LFR-bilinear-nominal} with the uncertainty characterization derived in Section~\ref{sec:LFR-nominal}.

\subsection{Gain-scheduling controller using the uncertainty}\label{sec:controller-gain-scheduling}
In this section, we employ gain-scheduling techniques~\cite{scherer:2001} to enhance the design in Theorem~\ref{thm:stability-condition-LFR-nominal} leading to a more flexible controller improving the feasibility and the closed-loop behavior. 
In particular, we design a controller which additionally depends on the uncertainty channel of the LFR via $w$, i.e., 
\begin{equation}\label{eq:controller-gain-scheduling} 
    u(z) = Kz + K_w w
\end{equation}
where $K\in\bbR^{m\times N}$ and $K_w\in\bbR^{m\times Nm}$. Substituting the input~\eqref{eq:controller-gain-scheduling} in~\eqref{eq:LFR-bilinear-open-loop}, we obtain the corresponding closed-loop LFR
\begin{subequations}\label{eq:LFR-bilinear-gain-scheduling}
    \begin{align}
        \begin{bmatrix}
            z_+ \\ u
        \end{bmatrix}
        &= \begin{bmatrix}
            A+B_0 K & \tB + B_0 K_w \\
            K & K_w
        \end{bmatrix}
        \begin{bmatrix}
            z \\ w
        \end{bmatrix},\\
        w &= (I_m\kron z)  u
    \end{align}
\end{subequations}
with $(I_m\kron z)\in\mathbf{\Delta}$. The following theorem establishes a controller design method guaranteeing stability of the closed-loop system~\eqref{eq:LFR-bilinear-gain-scheduling}.

\begin{theorem}\label{thm:stability-condition-LFR-gain-scheduling}
    If there exist a symmetric $N\times N$ matrix ${P}=P^\top\succ 0$, matrices ${L}\in\bbR^{m\times N}$, $L_w\in\bbR^{m\times Nm}$, a symmetric $m\times m$ matrix $\tilde{\Lambda}=\tilde{\Lambda}\succ 0$, and a scalar $\nu> 0$ such that $\cQ_\mathrm{GS}\succ 0$ and~\eqref{eq:stability-condition-LFR-nominal-invariance} holds, where
    \small
    \begin{equation*}
        \cQ_\mathrm{GS} = \cQ + 
        \begin{bmatrix}
            0
            & - B_0 L_w(I_m\kron \hS_z) 
            & 0
            & B_0 L_w
            \\
            \star
            & - L_w (I_m\kron\hS_z) - (I_m\kron\hS_z^\top) L_w^\top
            & 0
            & L_w
            \\
            \star 
            & \star
            & 0
            & 0
            \\
            \star 
            & \star
            & \star
            & 0
        \end{bmatrix}
    \end{equation*}
    \normalsize
    with $\hS_z = \tQ_z^{-1}\tS_z$,
    then $\cZ_\mathrm{RoA} = \{z\in\bbR^N \mid z^\top P^{-1} z \leq 1\}\subseteq \cZ$ and the controller
    \begin{equation}\label{eq:controller-gain-scheduling-explicit} 
        u(z) = (I - L_w(\tilde{\Lambda}^{-1} \kron \tQ_z^{-1} z))^{-1} LP^{-1} z
    \end{equation}
    locally asymptotically stabilizes system~\eqref{eq:dynamics-bilinear} for all initial conditions $z_0\in\cZ_\mathrm{RoA}$.
\end{theorem}
\begin{proof}
    We define $K = L P^{-1}$ and $K_w = L_w(\tilde{\Lambda}^{-1} \kron \tQ_z^{-1})$. Then, by using the Schur complement twice, $\cQ_\mathrm{GS}\succ 0$ is equivalent to 
    \begin{multline}\label{eq:proof-gain-scheduling-Schur}
        \begin{bmatrix}
            P
            & -(\tB + B_0 K_w) (\tilde{\Lambda}\kron\tS_z)
            \\
            \star
            &  \tilde{\Lambda}\kron\tR_z - K_w (\tilde{\Lambda}\kron\tS_z) - (\tilde{\Lambda}\kron\tS_z^\top) K_w^\top
        \end{bmatrix}
        \\
        + 
        \begin{bmatrix}
            \tB + B_0 K_w \\ K_w
        \end{bmatrix}
        (\tilde{\Lambda}\kron\tQ_z)
        \begin{bmatrix}
            \tB + B_0 K_w \\ K_w
        \end{bmatrix}^\top 
        \\
        - 
        \begin{bmatrix}
            A+B_0 K \\ K
        \end{bmatrix}
        P
        \begin{bmatrix}
            A+B_0 K \\ K
        \end{bmatrix}^\top 
        \succ 0.
    \end{multline}
    Note that 
    \begin{multline*}
        \begin{bmatrix}
            P & 0 \\ 0 & 0
        \end{bmatrix}
        - \begin{bmatrix}
            A+B_0 K \\ K
        \end{bmatrix}
        P
        \begin{bmatrix}
            A+B_0 K \\ K
        \end{bmatrix}^\top
        \\
        = \begin{bmatrix}
            \star
        \end{bmatrix}^\top
        \begin{bmatrix}
            -P & 0 \\ 0 & P
        \end{bmatrix}
        \begin{bmatrix}
            (A+B_0K)^\top & K^\top \\ -I & 0
        \end{bmatrix}
    \end{multline*}
    and
    \begin{multline*}
        \begin{bmatrix}
            0
            & - (\tB+B_0K_w) (\tilde{\Lambda}\kron\tS_z)
            \\
            \star
            & \tilde{\Lambda}\kron\tR_z - K_w (\tilde{\Lambda}\kron\tS_z) - (\tilde{\Lambda}\kron\tS_z^\top) K_w^\top
        \end{bmatrix}
        \\
        + \begin{bmatrix}
            \tB + B_0 K_w \\ K_w
        \end{bmatrix}
        (\tilde{\Lambda}\kron\tQ_z)
        \begin{bmatrix}
            \tB + B_0 K_w \\ K_w
        \end{bmatrix}^\top 
        \\
        = \begin{bmatrix}
           \star
        \end{bmatrix}^\top 
        \begin{bmatrix}
            \tilde{\Lambda}\kron\tQ_z & \tilde{\Lambda}\kron\tS_z \\ \tilde{\Lambda}\kron\tS_z^\top & \tilde{\Lambda}\kron\tR_z
        \end{bmatrix}
        \begin{bmatrix}
            (\tB + B_0K_w)^\top & K_w^\top \\ 0 & -I
        \end{bmatrix}.
    \end{multline*}
    Moreover, we recall~\eqref{eq:multiplier-trafo} and observe that $T$ is unitary, i.e., $T^{-1}=T^\top$. Then, we directly obtain 
    \begin{equation*}
        \Pi_{\Delta}^{-1} 
        = T \left(
            \tilde{\Lambda} \kron \begin{bmatrix} \tQ_z & \tS_z \\ \tS_z^\top & \tR_z \end{bmatrix}
        \right) T^\top
    \end{equation*}
    using $(U\kron V)^{-1} = U^{-1} \kron V^{-1}$. By the definition of $T$ we deduce that 
    \begin{equation*}
        \Pi_{\Delta}^{-1}
        = \begin{bmatrix}
            \tilde{\Lambda} \kron \tQ_z & \tilde{\Lambda} \kron \tS_z \\
            \tilde{\Lambda} \kron \tS_z^\top & \tilde{\Lambda} \kron \tR_z
        \end{bmatrix}.
    \end{equation*}
    Thus, we write~\eqref{eq:proof-gain-scheduling-Schur} equivalently as
    \begin{equation*}\small
        \begin{bmatrix}
            \star
        \end{bmatrix}^\top
        \left[\def\arraystretch{1.15}\begin{array}{cc|cc}
            -P & 0 & 0 & 0 \\ 
            0 & P & 0 & 0 \\\hline
            0 & 0 & \multicolumn{2}{c}{\multirow{2}{*}{$\Pi_{\Delta}^{-1}$}} \\
            0 & 0
        \end{array}\right]
        \left[\def\arraystretch{1.15}\begin{array}{cc}
            (A+B_0K)^\top & K^\top \\ -I & 0 \\\hline
            (\tB+B_0K_w)^\top & K_w^\top \\ 0 & -I
        \end{array}\right]
        \succ 0.
    \end{equation*}
    Using the dualization lemma~\cite[Lm. 4.9]{scherer:weiland:2000}, we obtain
    \begin{equation*}\small
        \begin{bmatrix}
            \star
        \end{bmatrix}^\top
        \left[\def\arraystretch{1.15}\begin{array}{cc|cc}
            -\tilde{P} & 0 & 0 & 0 \\
            0 & \tilde{P} & 0 & 0 \\\hline
            0 & 0 & \multicolumn{2}{c}{\multirow{2}{*}{$\Pi_{\Delta}$}}\\
            0 & 0
        \end{array}\right]
        \left[\def\arraystretch{1.15}\begin{array}{cc}
            I & 0 \\
            A+B_0K & \tB+B_0K_w \\\hline
            0 & I \\
            K & K_w
        \end{array}\right]
        \prec 0,
    \end{equation*}
    where $\tilde{P}=P^{-1}$, $\Lambda=\tilde{\Lambda}^{-1}$. 
    Hence, we define the Lyapunov function $V(z) = z^\top \tilde{P} z$ and conclude $\Delta V(z) = z_+^\top \tilde{P} z_+ - z^\top \tilde{P} z < 0$ for the LFR in~\eqref{eq:LFR-bilinear-gain-scheduling} for all $(I_m\kron z)\in\mathbf{\Delta}\setminus\{0\}$ and the controller~\eqref{eq:controller-gain-scheduling} due to~\cite[Thm. 2]{scherer:2001}.
    Thus, the obtained controller guarantees $\Delta V(z) < 0$ for the bilinear system~\eqref{eq:dynamics-bilinear} for all $z\in\cZ\setminus\{0\}$ due to Proposition~\ref{prop:multiplier-tightness}.

    To show asymptotic stability for all $z\in\cZ_\mathrm{RoA}$, it remains to show $\cZ_\mathrm{RoA}\subseteq \cZ$ and positive invariance of $\cZ_\mathrm{RoA}$, i.e., $z_+\in\cZ_\mathrm{RoA}$ if $z\in\cZ_\mathrm{RoA}$. 
    For the set inclusion $\cZ_\mathrm{RoA}\subseteq\cZ$, note that~\eqref{eq:stability-condition-LFR-nominal-invariance} is equivalent to 
    \begin{equation*}
        \begin{bmatrix}\star\end{bmatrix}^\top
        \left[\def\arraystretch{1.15}\begin{array}{cc|cc}
            \tQ_z & \tS_z & 0 & 0 \\
            \tS_z^\top & \tR_z & 0 & 0 \\\hline
            0 & 0 & \frac{1}{\nu} P & 0 \\
            0 & 0 & 0 & -\frac{1}{\nu}
        \end{array}\right]
        \left[\def\arraystretch{1.15}\begin{array}{cc}
            -I & 0 \\
            0 & I \\\hline
            I & 0 \\
            0 & -I
        \end{array}\right]
        \preceq 0.
    \end{equation*}
    Using again the dualization lemma, this is equivalent to 
    \begin{equation*}
        \begin{bmatrix}
            Q_z & S_z \\ S_z^\top & R_z
        \end{bmatrix}
        - \nu \begin{bmatrix}
            - \tilde{P} & 0 \\ 0 & 1
        \end{bmatrix}
        \succeq 0.
    \end{equation*}
    Then, multiplying from left and right by $\begin{bmatrix}
        z^\top & 1
    \end{bmatrix}^\top$ and its transpose, respectively, and applying the S-procedure (cf.~\cite{scherer:weiland:2000,boyd:vandenberghe:2004}) results in $z\in\cZ$ for all $z\in\cZ_\mathrm{RoA}$, i.e., $\cZ_\mathrm{RoA}\subseteq \cZ$.    
    Positive invariance of $\cZ_\mathrm{RoA}$ can be directly deduced from the choice of $\cZ_\mathrm{RoA}$ as a sublevel set of the Lyapunov function $V(z)$ and $\Delta V(z)\leq 0$ for $z\in\cZ\supseteq\cZ_\mathrm{RoA}$. 
    Hence, we conclude that the obtained controller asymptotically stabilizes the bilinear system~\eqref{eq:dynamics-bilinear} for all $z\in\cZ_\mathrm{RoA}$. 
    
    Moreover, we show that the controller $u(z)$ is indeed of the form given in~\eqref{eq:controller-gain-scheduling-explicit}, i.e., it holds that
    \begin{align*}
        u(z) &= K z + K_w w(z)
        \;\overset{\eqref{eq:LFR-bilinear-gain-scheduling}}{{=}}\; Kz + K_w (I_m \kron z) u(z) 
        \\
        &= (I - K_w(I_m\kron z))^{-1} K z 
        \\
        &= (I - L_w(\tilde{\Lambda}^{-1} \kron \tQ_z^{-1}z))^{-1} LP^{-1} z 
        \\
        &\eqqcolon K_\mathrm{new}(z) z,
    \end{align*}
    where $K_\mathrm{new}: \bbR^N \to \bbR^{m\times N}$. We note that $(I - K_w(I_m\kron z))$ is non-singular for all $(I_m\kron z)\in\mathbf{\Delta}$ due to~\cite[Thm.~2]{scherer:2001} and, using Proposition~\ref{prop:multiplier-tightness}, for all $z\in\cZ_\mathrm{RoA}$.
\end{proof}
The design procedure in Theorem~\ref{thm:stability-condition-LFR-gain-scheduling} yields a more flexible state-feedback controller~\eqref{eq:controller-gain-scheduling-explicit} which is a rational function in the state $z$. 
The proposed approach uses gain-scheduling to exploit that we can measure the uncertainty $I_m\kron z$ in the LFR~\eqref{eq:LFR-bilinear-nominal}. 
Solving the resulting LMI feasibility condition yields a linear parameter-varying controller which enhances the resulting closed-loop behavior by reducing conservatism in the LFR of the bilinear system in comparison to a linear state feedback based on~\cite[Thm.~1]{khlebnikov:2018} ($m=1$) or Theorem~\ref{thm:stability-condition-LFR-nominal} ($m\geq 1$).
Note that the choice $L_w=0$ always reduces the controller to the \emph{linear} state-feedback case, i.e., $\cQ_\mathrm{GS}=\cQ$ and $u(z) = LP^{-1}z$. 
Thus, whenever the LMIs in Theorem~\ref{thm:stability-condition-LFR-nominal} (linear state-feedback) are feasible, the LMIs in Theorem~\ref{thm:stability-condition-LFR-gain-scheduling} (gain-scheduling) are feasible as well.
\begin{remark}
    Note that the controller~\eqref{eq:controller-gain-scheduling} is exposed to a \emph{known} uncertainty, namely the state $z$. 
    Controllers as in~\eqref{eq:controller-gain-scheduling} are commonly referred to as full-information feedback controllers (cf.~\cite{doyle:glover:khargonekar:francis:1989,packard:zhou:pandey:leonhardson:balas:1992,astolfi:1997}), being specialized forms of gain-scheduling controllers which can handle \emph{unknown} uncertainties in general. 
    In particular, the controller has access to both the state and the uncertainty.
    An interesting direction for future research contains the generalization to unknown uncertainties.
\end{remark}

\section{QUADRATIC PERFORMANCE}\label{sec:controller-performance}
Next, we include a performance goal in the controller design which needs to be satisfied by the resulting closed-loop system. In particular, we add a performance channel $w_p \mapsto z_p$ to system~\eqref{eq:dynamics-bilinear} and consider 
\begin{subequations}\label{eq:dynamics-bilinear-performance}
    \begin{align}
        z_+ &= Az + B_0 u + \tB(u\kron z) + B_{p} w_p, \\
        z_p &= C_p z + D_{p,u} u + \tD_{p,uz} (u\kron z) + D_{p,w} w_p,
    \end{align}
\end{subequations}
where $z_p\in\bbR^{p}$ and $w_p\in\bbR^{q}$.
Note that the performance output $z_p$ can depend bilinearly on the state $z$ and the input $u$, and linearly on the performance input $w_p$. As before, we write~\eqref{eq:dynamics-bilinear-performance} in closed loop with the gain-scheduling controller~\eqref{eq:controller-gain-scheduling} as the LFR%
\begin{subequations}\label{eq:LFR-bilinear-performance}
    \begin{align}
        \begin{bmatrix}
            z_+ \\ u \\ z_p
        \end{bmatrix}
        &= \begin{bmatrix}
            \cA & \cB & B_{p}\\
            K & K_w & 0 \\
            \cC & \cD & D_{p,w}
        \end{bmatrix}
        \begin{bmatrix}
            z \\ w \\ w_p
        \end{bmatrix},\\
        w &= (I_m\kron z) u
    \end{align}
\end{subequations}%
with $\cA = A+B_0K$, $\cB = \tB + B_0K_w$, $\cC = C_p + D_{p,u}K$, $\cD = \tD_{p,uz} + D_{p,u}K_w$, and $(I_m\kron z)\in\mathbf{\Delta}$. In the following, we consider a local quadratic performance specification.
\begin{definition}\label{def:quadratic-performance}
    The closed-loop system~\eqref{eq:LFR-bilinear-performance} satisfies local quadratic performance with index $
        \Pi_p = \begin{bmatrix}
            Q_p & S_p \\ S_p^\top & R_p
        \end{bmatrix}
    $ and supply rate 
    \begin{equation*}
        s(w_p,z_p) 
        = \begin{bmatrix}
            w_p \\ z_p    
        \end{bmatrix}^\top 
        \Pi_p
        \begin{bmatrix}
            w_p \\ z_p    
        \end{bmatrix},
    \end{equation*}
    where $R_p\succeq 0$, $Q_p \prec 0$, if there exist $\varepsilon,\delta>0$ such that 
    \begin{equation}\label{eq:quadratic-performance}
        \sum_{k=0}^\infty 
        s(w_{p,k},z_{p,k})
        \leq 
        - \varepsilon
        \sum_{k=0}^\infty 
        \|w_{p,k}\|^2
    \end{equation}
    for all $w_p\in\cB_\delta=\{w_p\in\bbR^q \mid \|w_p\|^2 \leq \delta\}$.
\end{definition}
For instance, $Q_p = -\gamma^2 I$, $S_p = 0$, and $R_p=I$ correspond to a local $\cL_2$-gain bound $\gamma$ on the performance channel. 
\begin{assumption}\label{ass:lower-bound-suppy-rate}
    There exists a continuous and strictly increasing function $\alpha: [0,\infty) \to [0,\infty)$ with $\alpha(0)=0$ such that $s(w_p,z_p) \geq - \alpha(\|w_p\|^2)$.
\end{assumption}
Assumption~\ref{ass:lower-bound-suppy-rate} is satisfied for commonly used supply rates $s(w_p,z_p)$, e.g., the supply rate corresponding to an $\cL_2$-gain bound satisfies $s(w_p,z_p)\geq -\gamma^2\|w_p\|^2$. 
Note that we consider \emph{local} performance since our analysis relies on invariance of $z\in\cZ$. 
Under the presence of arbitrary disturbances $w_p$, this invariance might be violated. As a remedy, we assume boundedness of $w_p$ and show later robust positive invariance of $\cZ_\mathrm{RoA}\subseteq\cZ$ to ensure invariance of $(I_m \kron z) \in \mathbf{\Delta}$ in the LFR~\eqref{eq:LFR-bilinear-performance}. 

Further, we assume that $\Pi_p$ is invertible and define 
\begin{equation*}
    \begin{bmatrix}
        \tQ_p & \tS_p \\
        \tS_p^\top & \tR_p
    \end{bmatrix}
    =
    \begin{bmatrix}
        Q_p & S_p \\
        S_p^\top & R_p
    \end{bmatrix}^{-1}.
\end{equation*}

\begin{theorem}\label{thm:stability-condition-LFR-performance}
    Suppose Assumption~\ref{ass:lower-bound-suppy-rate} holds. If there exist a symmetric $N\times N$ matrix ${P}=P^\top\succ 0$, matrices ${L}\in\bbR^{m\times N}$, $L_w\in\bbR^{m\times Nm}$, a symmetric $m\times m$ matrix $\tilde{\Lambda}=\tilde{\Lambda}\succ 0$, and scalars $\nu>0,\tilde{\lambda}>0$ such that~\eqref{eq:stability-condition-LFR-performance} and~\eqref{eq:stability-condition-LFR-nominal-invariance} hold,%
    \begin{figure*}[!t]
        \begin{equation}\label{eq:stability-condition-LFR-performance}
            \left[
                \def\arraystretch{1.15}\begin{array}{cccc|c}
                    \multicolumn{4}{c|}{\multirow{4}{*}{$\cQ_\mathrm{GS}$}}
                    & \tilde{\lambda} B_p(\tQ_p D_{p,w}^\top - \tS_p)
                    \\
                    & & &
                    & -(\tD_{p,uz}(\tilde{\Lambda}\kron\tS_z) + D_{p,u}L_w(I_m\kron\tQ_z^{-1}\tS_z))^\top
                    \\
                    & & &
                    & (C_p P + D_{p,u}L)^\top
                    \\
                    & & &
                    & (\tD_{p,uz}(\tilde{\Lambda}\kron I_N) + D_{p,u}L_w)^\top 
                    \\\hline
                    \star & \star & \star & \star
                    & \tilde{\lambda}(\tR_p -D_{p,w}\tS_p - \tS_p^\top D_{p,w}^\top + D_{p,w} \tQ_p D_{p,w}^\top)
                \end{array}
            \right]
            + \tilde{\lambda}\begin{bmatrix}
                B_p \\ 0 \\ 0 \\ 0 \\ 0
            \end{bmatrix}
            \tQ_p
            \begin{bmatrix}
                B_p \\ 0 \\ 0 \\ 0 \\ 0
            \end{bmatrix}^\top
            \succ 0
        \end{equation}
        \medskip
        \hrule
    \end{figure*}
    then $\cZ_\mathrm{RoA} = \{z\in\bbR^N \mid z^\top P^{-1} z \leq 1\}\subseteq \cZ$ and the controller~\eqref{eq:controller-gain-scheduling-explicit} achieves local asymptotic stability and quadratic performance of system~\eqref{eq:dynamics-bilinear-performance} for all initial conditions $z_0\in\cZ_\mathrm{RoA}$.
\end{theorem}
\begin{proof}
    We divide the proof into two parts, where we first establish robust positive invariance of $\cZ_\mathrm{RoA}\subseteq\cZ$ and then show local quadratic performance.
    
    \paragraph*{Part I: Robust positive invariance of $\cZ_\mathrm{RoA}\subseteq\cZ$}
    First, note that the set inclusion $\cZ_\mathrm{RoA}\subseteq\cZ$ is similarly obtained as in the proof of Theorem~\ref{thm:stability-condition-LFR-gain-scheduling}.
    Further, the set $\cZ_\mathrm{RoA}$ is robust positively invariant if there exists a $\delta > 0$ such that $z_+\in\cZ_\mathrm{RoA}$ for all $z\in\cZ_\mathrm{RoA}$ and $w_p\in\cB_\delta$.
    According to the definition of $\cZ_\mathrm{RoA}$, we define the function $V(z) = z^\top P^{-1} z$ such that robust positive invariance is equivalently described by $V(z_+)\leq V(z)$ for all $z\in\cZ_\mathrm{RoA}$ and $w_p\in\cB_\delta$.
    In the following, $V(z)$ will serve as a Lyapunov function.
    Since the presented controller design procedure is framed in the robust control framework of~\cite{scherer:2000}, the following steps are a direct adaption of the proof of Theorem~\ref{thm:stability-condition-LFR-gain-scheduling}. 
    
    In particular, given~\eqref{eq:stability-condition-LFR-performance}, we use the Schur complement twice and apply the dualization lemma~\cite[Lm.~4.9]{scherer:weiland:2000} to arrive at $\Xi \prec 0$, where
    \begin{equation*}\small
        \Xi = \begin{bmatrix}
            \star
        \end{bmatrix}^\top
        \left[\def\arraystretch{1.15}\begin{array}{cc|cc|cc}
            -\tilde{P} & 0 & 0 & 0 & 0 & 0 \\ 
            0 & \tilde{P} & 0 & 0 & 0 & 0 \\\hline
            0 & 0 & \multicolumn{2}{c|}{\multirow{2}{*}{$\Pi_{\Delta}$}} & 0 & 0 \\
            0 & 0 & & & 0 & 0 \\\hline
            0 & 0 & 0 & 0 & \multicolumn{2}{c}{\multirow{2}{*}{$\lambda\Pi_p$}}\\
            0 & 0 & 0 & 0
        \end{array}\right]
        \left[\def\arraystretch{1.15}\begin{array}{ccc}
            I & 0 & 0 \\ \cA & \cB & B_p \\\hline
            0 & I & 0 \\ K & K_w & 0 \\\hline 
            0 & 0 & I \\ \cC & \cD & D_{p,w}
        \end{array}\right]
    \end{equation*}
    with $K=LP^{-1}$, $K_w=L_w(\tilde{\Lambda}^{-1}\kron \tQ_z^{-1})$, $\tilde{P}=P^{-1}$, $\Lambda=\tilde{\Lambda}^{-1}$, $\lambda=\tilde{\lambda}^{-1}$.
    Since the inequality $\Xi\prec 0$ is strict, there exist $\rho,\varepsilon>0$ such that 
    \begin{equation}\label{eq:proof-performance-Xi-nonstrict}
        \Xi + \diag(\rho I,0,\varepsilon I) \preceq 0.    
    \end{equation}
    Recall $V(z) = z^\top \tilde{P} z$ and the definition of the uncertainty characterization $\mathbf{\Delta}$.
    Then, multiplying~\eqref{eq:proof-performance-Xi-nonstrict} from left and right by $\begin{bmatrix}
        z^\top & w^\top & w_p^\top
    \end{bmatrix}^\top$ and its transpose, respectively, yields
    \begin{equation}\label{eq:proof-performance-dissipation-inequality}
        \Delta V(z) \leq - (\rho \|z\|^2 + \varepsilon\|w_p\|^2 + \lambda s(w_p,z_p))
    \end{equation}
    for all $(I_m\kron z)\in\mathbf{\Delta}$ (compare~\cite[Thm.~10.4]{scherer:2000}).
    Hence, the obtained controller guarantees~\eqref{eq:proof-performance-dissipation-inequality} for the bilinear system~\eqref{eq:dynamics-bilinear-performance} for all $z\in\cZ$ due to Proposition~\ref{prop:multiplier-tightness}.
    Since $\cZ_\mathrm{RoA}\subseteq \cZ$, this holds in particular for all $z\in\cZ_\mathrm{RoA}$.

    Next, we observe
    \begin{equation*}
        \rho \|z\|^2 + \varepsilon\|w_p\|^2 + \lambda s(w_p,z_p)
        \geq \rho \|z\|^2 - \lambda \alpha(\|w_p\|^2),
    \end{equation*}
    where we use the lower bound for $s(w_p,z_p)$ in Assumption~\ref{ass:lower-bound-suppy-rate}. 
    Thus,~\eqref{eq:proof-performance-dissipation-inequality} leads to 
    \begin{equation*}
        V(z_+) \leq V(z) -\rho \|z\|^2 + \lambda \alpha(\|w_p\|^2).
    \end{equation*}
    Using $V(z)\leq \|z\|^2 \|\tP\|_2$ or, in particular, $\|z\|^2 \geq \frac{V(z)}{\|\tP\|_2}$, yields
    \begin{equation*}
        V(z_+) \leq \left(1-\frac{\rho}{\|\tP\|_2}\right) V(z) + \lambda \alpha(\|w_p\|^2).
    \end{equation*}
    Hence, for all $z\in\cZ_\mathrm{RoA}$, i.e., $V(z)\leq 1$, we obtain $V(z_+)\leq 1$ for all $w_p$ with $\|w_p\|^2 \leq \delta \coloneqq \alpha^{-1}\left(\frac{\rho}{\lambda\|\tP\|_2}\right)$.
    Thus, we have established the existence of a $\delta>0$ such that the region of attraction $\cZ_\mathrm{RoA}$ is robust positively invariant.
     
    \paragraph*{Part II: Local quadratic performance}
    Due to the established robust positive invariance of $\cZ_\mathrm{RoA}$, we know that~\eqref{eq:proof-performance-dissipation-inequality} holds for all times.
    Then, by building the sum of~\eqref{eq:proof-performance-dissipation-inequality} for all times $k=0$ to $k\to\infty$, we establish~\eqref{eq:quadratic-performance} and, thus, quadratic performance according to Definition~\ref{def:quadratic-performance}.
    Hence, LMIs~\eqref{eq:stability-condition-LFR-performance},~\eqref{eq:stability-condition-LFR-nominal-invariance} ensure local asymptotic stability and quadratic performance of the bilinear system~\eqref{eq:dynamics-bilinear-performance} for all initial conditions $z_0\in\cZ_\mathrm{RoA}$. 
\end{proof}

Similar to the discussion after Theorem~\ref{thm:stability-condition-LFR-gain-scheduling},~\eqref{eq:stability-condition-LFR-performance} reduces for $L_w=0$ to a sufficient condition for local closed-loop stability and quadratic performance for all initial conditions $z_0\in\cZ_\mathrm{RoA}$ under linear state feedback. Hence, the achieved performance of the rational controller obtained by the design procedure in Theorem~\ref{thm:stability-condition-LFR-performance} is at least as good as the performance using linear state feedback.  
In order to obtain closed-loop guarantees for the largest possible region $\cZ_\mathrm{RoA}\subseteq\cZ$, we optimize its size. To this end, we maximize $\tr(P)$ such that~\eqref{eq:stability-condition-LFR-nominal-invariance} and~\eqref{eq:stability-condition-LFR-performance} hold.

Due to the considered LFR framework, we note that additional uncertainties acting on the bilinear system~\eqref{eq:dynamics-bilinear-performance} can be easily included in the proposed controller design approach, which is a crucial advantage over the framework presented in~\cite{khlebnikov:2018}.
One direction could build on~\cite{strasser:berberich:allgower:2023}, where a similar LFR-based framework is used for data-driven control of general nonlinear systems with scalar inputs. 
Here, Koopman operator theory is used to deduce an equivalent error-affected bilinear representation of the unknown nonlinear system leading to a controller design with closed-loop guarantees for the true system based on the bilinear description. 
The design approach proposed in the present paper may allow to solve such nonlinear data-driven control problems in a possibly less conservative way.

\section{SIMULATION EXAMPLES}\label{sec:simulation-examples}
In the following, we illustrate the benefits of our approach by means of a numerical example. All simulations are conducted in Matlab using the toolbox YALMIP~\cite{lofberg:2004} with the semidefinite programming solver MOSEK~\cite{mosek:2022}.

\begin{example}\label{exmp:bilinear-scalar}
    Consider the scalar bilinear system 
    \begin{equation*}
        z_+ = z + (z+1)u
    \end{equation*}
    in the local region $\cZ=\{z\in\bbR\mid z^2\leq R_z\}$, i.e., $Q_z = -1$, $S_z=0$, and a variable $R_z\geq 0$ leading to different sizes of $\cZ$.
    At first, we are only interested in closed-loop stability without explicit performance objective.
    To obtain the largest subset of the region of attraction $\cZ_\mathrm{RoA}=\{z\in\bbR\mid z^2\leq P\}\subseteq\cZ$, we maximize $\tr(P)$ subject to~\eqref{eq:stability-condition-LFR-nominal-invariance} and $\cQ\succ 0$ or $\cQ_\mathrm{GS}\succ 0$ for the controller designs of Theorem~\ref{thm:stability-condition-LFR-nominal} or Theorem~\ref{thm:stability-condition-LFR-gain-scheduling}, respectively. Choosing $R_z = 0.9$, both optimizations yield $P = 0.9$, i.e., $\cZ_\mathrm{RoA}=\cZ$ with the respective linear state-feedback controller 
    \begin{equation*}
        u(z) = -0.6178 z
    \end{equation*}
    and the gain-scheduling enhanced controller
    \begin{equation*}
        u(z) = -\frac{0.5324z}{1+0.5762z}.
    \end{equation*}
    In particular, both approaches lead to a stabilizing controller in the local region $\cZ_\mathrm{RoA}=\{z\in\bbR\mid z^2<1\}$. For larger $\cZ_\mathrm{RoA}$, the design procedures are not able to find a stabilizing controller as then the uncertainty region includes $z=-1$ for which the system is not controllable. 
    
    To further analyze the proposed controller design, we add a performance channel, i.e., we consider $
        z_+ = z + (z+1)u + w_p
    $ and the performance channel $w_p\mapsto z_p$, where $z_p=z$. The desired performance goal is to minimize the local $\cL_2$-gain of the corresponding performance channel. To this end, we choose the performance index $Q_p = -\gamma^2$, $S_p = 0$, $R_p = 1$ and compare the closed-loop behavior of both controllers for two scenarios.
    
    1) First, we consider again the region $\cZ=\{z\in\bbR \mid z^2\leq 0.9\}$ and search for the minimal achievable $\cL_2$-gain for the region $\cZ_\mathrm{RoA}\subseteq\cZ$, where we again maximize $\tr(P)$ subject to the design conditions. 
    Then, on the one hand, linear state feedback (i.e., Theorem~\ref{thm:stability-condition-LFR-performance} with $L_w=0$) achieves an $\cL_2$-gain bound of $\gamma=19.49$ for $z\in\cZ_\mathrm{RoA}=\cZ$ and the controller $u(z)= -z$ with the resulting closed-loop system $z_+ = z + (z+1)u(z) + w_p = -z^2 + w_p$. On the other hand, designing a gain-scheduling enhanced controller based on Theorem~\ref{thm:stability-condition-LFR-performance} for general $L_w$ leads to the $\cL_2$-gain bound of $\gamma=1.001$ for $z\in\cZ_\mathrm{RoA}=\cZ$ and the controller $u(z) = -\frac{z}{1+z}$ with the resulting closed loop $z_+ = z + (z+1)u(z) + w_p = w_p$. 
    Thus, the additional flexibility of the latter controller benefits the closed-loop behavior and improves the achieved $\cL_2$-performance compared to linear state feedback.
    
    2) Next, we compare the largest possible region of attraction $\cZ_\mathrm{RoA}$ such that an $\cL_2$-gain of $\gamma=1.5$ can be guaranteed for all $z\in\cZ_\mathrm{RoA}$ using both methods. While the state-feedback approach leads to a stabilizing controller with the desired $\cL_2$-gain for all $z^2\leq 0.1111$, the gain-scheduling controller design guarantees the $\cL_2$-gain for $z^2\leq 0.9999$.

    Figure~\ref{fig:feasibility-comparison} shows clearly the improvements of the proposed gain-scheduling controller by comparing the achievable $\cL_2$-gain bound for different sizes of the region of attraction $\cZ_\mathrm{RoA}$ using a linear state-feedback controller and a gain-scheduling controller.
    We note that a lower bound on the \emph{true} $\cL_2$-gain is obtained via uniform sampling of disturbance signals $w_p\in\cB_\delta$ and simulating the closed loop. 
    In particular, we compute the empirical lower bound $\gamma_\mathrm{lb}=1$ via the worst case amplification for a sample size of $10^7$. 
    Thus, the analysis condition~\eqref{eq:stability-condition-LFR-performance} evaluated for the closed loop with the proposed gain-scheduling controller~\eqref{eq:controller-gain-scheduling-explicit} leads indeed to a tight $\cL_2$-gain for this example system.

    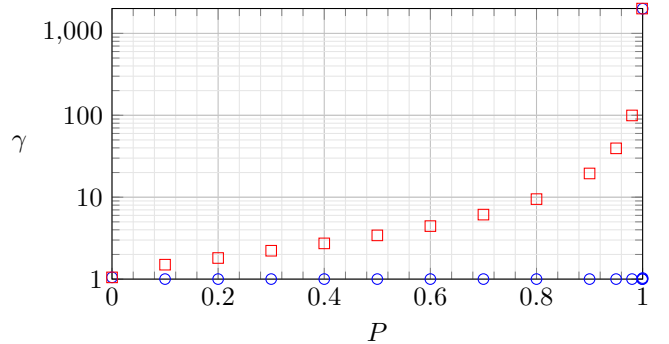
\begin{figure}[htb]
        \centering
        \begin{tikzpicture}[%
            /pgfplots/every axis y label/.style={at={(0,0.5)},xshift=-35pt,rotate=0},%
            /pgfplots/every axis x label/.style={at={(0.5,0)},yshift=-20pt,rotate=0},%
          ]%
            \begin{axis}[
                ymode=log,
                log ticks with fixed point,
                ymin=1,ymax=2010,
                ylabel={$\gamma$},
                xmin=0,xmax=1,
                xlabel={$P$},
                minor x tick num=4,
                grid=both,
                minor grid style={gray!20},
                height = 0.6\columnwidth,
                width=\columnwidth,
            ]  
                \addplot[color=blue,only marks,mark=o] coordinates{(0.0001,1.05) (0.1,1.001) (0.2,1.001) (0.3,1.001) (0.4,1.001) (0.5,1.001) (0.6,1.001) (0.7,1.001) (0.8,1.001) (0.9,1.001) (0.95,1.001) (0.98,1.001) (0.999,1.001) (0.99995,1.001) (0.9999999,1.03) (0.999999999108,2000.32)};
                \addplot[color=red,only marks,mark=square] coordinates{(0.0001,1.05) (0.1,1.5) (0.2,1.81) (0.3,2.22) (0.4,2.73) (0.5,3.42) (0.6,4.44) (0.7,6.13) (0.8,9.48) (0.9,19.49) (0.95,39.5) (0.98,99.52) (0.999,2007.54)};
            \end{axis}
        \end{tikzpicture}
        \caption{Achievable $\cL_2$-gain bound $\gamma$ for different $P$ defining the region of attraction $\cZ_\mathrm{RoA}=\{z\in\bbR \mid z^2\leq P\}$ using linear state feedback ({\color{red}$\square$}) and a gain-scheduling controller ({\color{blue}$\circ$}).}
        \label{fig:feasibility-comparison}
    \end{figure}
\end{example}

\begin{example}
    Next, we want to apply the proposed design method to a practically relevant system. 
    For this reason, consider the cellular model
    \begin{align*}
        z_+^1 &= z^1 + T_\mathrm{s} z^2, \\
        z_+^2 &= z^2 + T_\mathrm{s} (a z^2 + u(az^1 + z^2) - cau),
    \end{align*}
    for a cattle that results from discretization with sampling rate $T_\mathrm{s}=0.01$ from the continuous-time dynamics proposed in~\cite{mohler:1973} and references therein. 
    Here, $z^1$ denotes the weight of the cattle, $z^2$ is the change of its weight, $a=13$ is a constant, and $c=0.6$ is a hereditary constant. 
    Moreover, the input $u$ is the total energy consumed minus energy required for maintenance.
    Further, we consider the discretized bilinear system in the local region $\cZ = \{z\in\bbR^{2} \mid z^\top z \leq R_z \}$. For this system, we want to maximize the region $\cZ_\mathrm{RoA}\subseteq \cZ$ for which the controller design procedure in Theorem~\ref{thm:stability-condition-LFR-gain-scheduling} guarantees closed-loop stability. For linear state feedback, i.e., $L_w=0$, we obtain the margin $R_z=0.28$ and a region of attraction 
    \begin{equation*}
        \cZ_\mathrm{RoA}=\left\{ 
            z\in\bbR^2
            \middle| 
            z^\top 
            \begin{bmatrix}
                3.61 & 0.31 \\ 0.31 & 6.04
            \end{bmatrix}
            z
            \leq 1
        \right\}
        \subset \cZ,
    \end{equation*}
    and, in particular, $\cZ_\mathrm{RoA}\neq \cZ$.
    For comparison, we also implement the controller design in~\cite[Thm.~1]{khlebnikov:2018} which depends nonlinearly on a design parameter $\epsilon$ affecting both feasibility and the closed-loop behavior. For the particular choice of $\epsilon=10^{-1}$, the resulting controller yields a smaller stability region than $\cZ_\mathrm{RoA}\subset \cZ$. We note that the controller design for a different choice of $\epsilon$ leads, if feasible, to an even smaller stability region.
    On the other hand, a gain-scheduling controller~\eqref{eq:controller-gain-scheduling-explicit} with $L_w\neq 0$ can guarantee closed-loop stability for the full region $\cZ$. 
    Moreover, we can enlarge its region of attraction further to the set $\cZ_\mathrm{GS} = \{z\in\bbR^2\mid z^\top z \leq 0.35\}$, demonstrating the merit of the proposed approach.
\end{example}

\begin{example}
    Finally, we illustrate the benefits of the multiplier class~\eqref{eq:Delta-representation-Pi-structured} for describing bilinear systems with multiple inputs. To this end, consider the system
    \begin{multline*}
        z_+ = \begin{bmatrix}
            1 & 1 & 1 \\
            1 & 0 & 1 \\
            1 & 0 & 1 
        \end{bmatrix}
        z 
        + \begin{bmatrix}
            1 & 0 \\ 1 & -1 \\ -1 & 1
        \end{bmatrix}
        u
        \\
        + \begin{bmatrix}
             1 & 0 &  1   & 1 & 0 & 0 \\
             0 & 1 &  1   & 0 & 1 & 0 \\
            -1 & 1 & -1   & 0 & 0 & 1
        \end{bmatrix}
        (u\kron z)
    \end{multline*}
    with a two-dimensional input for $\cZ=\{z\in\bbR^{3}\mid \|z\|\leq 0.1\}$.
    For this system, we design a gain-scheduling controller based on Theorem~\ref{thm:stability-condition-LFR-gain-scheduling} using both a full multiplier $0\prec \tilde{\Lambda}\in\bbR^{m\times m}$ and a diagonally repeated multiplier $\tilde{\Lambda} = \mu I_m$, where $\mu >0$. A full multiplier leads to a minimum $\cL_2$-gain bound of $\gamma=3.88$, whereas a gain-scheduling controller based on diagonally repeated multipliers achieves only $\gamma=4.13$.
\end{example}

\section{CONCLUSION}\label{sec:conclusion}
In this paper, we presented a controller design method with stability and performance guarantees for general multi-dimensional discrete-time bilinear systems. 
To this end, we represented the bilinear dynamics as an LFR interpreting the bilinearity as a structured uncertainty for which we constructed a suitable class of multipliers.
We then used this representation to derive an LMI-based design procedure guaranteeing local asymptotic stability and quadratic performance for the bilinear system. 
While the proposed approach contains linear state-feedback design as a special case, the gain-scheduling techniques allow for more general feedback controllers depending rationally on the state. 
The gain-scheduling enhanced controller is able to achieve better performance or a larger region for which desired closed-loop properties can be guaranteed.
These benefits were finally demonstrated in simulation for different numerical examples.

Interesting future work lies in the generalization of the design to more general rational controllers.
While such an extension could be pursued directly based on sum-of-squares methods~\cite{vatani:hovd:olaru:2014}, we see potential in our proposed uncertainty characterization to use its tightness to ease the computational burden of direct sum-of-squares techniques.
Moreover, we aim at the application of the proposed LFR approach to a data-driven setting building on results in~\cite{bisoffi:depersis:tesi:2020a,strasser:berberich:allgower:2021,guo:depersis:tesi:2021}.
The relevance of the developed approach is also strengthened by the recent widespread use of Koopman operator theory, which allows to represent a nonlinear system as a bilinear system, see, e.g.,~\cite{strasser:berberich:allgower:2023,sinha:nandanoori:drgona:vrabie:2022}.
Thus, an interesting direction is to develop a controller design with closed-loop guarantees for general nonlinear systems.

\bibliographystyle{IEEEtran}
\bibliography{literature}

\begin{thebibliography}{10}
\providecommand{\url}[1]{#1}
\csname url@samestyle\endcsname
\providecommand{\newblock}{\relax}
\providecommand{\bibinfo}[2]{#2}
\providecommand{\BIBentrySTDinterwordspacing}{\spaceskip=0pt\relax}
\providecommand{\BIBentryALTinterwordstretchfactor}{4}
\providecommand{\BIBentryALTinterwordspacing}{\spaceskip=\fontdimen2\font plus
\BIBentryALTinterwordstretchfactor\fontdimen3\font minus
  \fontdimen4\font\relax}
\providecommand{\BIBforeignlanguage}[2]{{%
\expandafter\ifx\csname l@#1\endcsname\relax
\typeout{** WARNING: IEEEtran.bst: No hyphenation pattern has been}%
\typeout{** loaded for the language `#1'. Using the pattern for}%
\typeout{** the default language instead.}%
\else
\language=\csname l@#1\endcsname
\fi
#2}}
\providecommand{\BIBdecl}{\relax}
\BIBdecl

\bibitem{khalil:2002}
H.~K. Khalil, \emph{Nonlinear systems}, 3rd~ed.\hskip 1em plus 0.5em minus
  0.4em\relax Upper Saddle River, NJ: Prentice-Hall, 2002.

\bibitem{mohler:1973}
R.~R. Mohler, \emph{Bilinear control processes: with applications to
  engineering, ecology and medicine}.\hskip 1em plus 0.5em minus 0.4em\relax
  Academic Press, Inc., 1973.

\bibitem{bruni:dipillo:koch:1974}
C.~Bruni, G.~DiPillo, and G.~Koch, ``Bilinear systems: An appealing class of
  "nearly linear" systems in theory and applications,'' \emph{IEEE Transactions
  on Automatic Control}, vol.~19, no.~4, pp. 334--348, 1974.

\bibitem{spinu:athanasopulos:lazar:bitsoris:2012}
V.~Spinu, N.~Athanasopoulos, M.~Lazar, and G.~Bitsoris, ``Stabilization of
  bilinear power converters by affine state feedback under input and state
  constraints,'' \emph{IEEE Transactions on Circuits and Systems II: Express
  Briefs}, vol.~59, no.~8, pp. 520--524, 2012.

\bibitem{espana:landau:1978}
M.~Espana and I.~D. Landau, ``Reduced order bilinear models for distillation
  columns,'' \emph{Automatica}, vol.~14, no.~4, pp. 345--355, 1978.

\bibitem{carleman:1932}
T.~Carleman, ``{Application de la théorie des équations intégrales
  linéaires aux systèmes d'équations différentielles non linéaires},''
  \emph{Acta Mathematica}, vol.~59, pp. 63 -- 87, 1932.

\bibitem{koopman:1931}
B.~O. Koopman, ``Hamiltonian systems and transformation in {H}ilbert space,''
  \emph{Proc. of the National Academy of Sciences of the United States of
  America}, vol.~17, no.~5, p. 315, 1931.

\bibitem{pedrycz:1980}
W.~Pedrycz, ``Stabilization of bilinear systems by a linear feedback control,''
  \emph{Kybernetika}, vol.~16, no.~1, pp. 48--53, 1980.

\bibitem{longchamp:1980}
R.~Longchamp, ``Stable feedback control of bilinear systems,'' \emph{IEEE
  Transactions on Automatic Control}, vol.~25, no.~2, pp. 302--306, 1980.

\bibitem{gutman:1981}
P.-O. Gutman, ``Stabilizing controllers for bilinear systems,'' \emph{IEEE
  Transactions on Automatic Control}, vol.~26, no.~4, pp. 917--922, 1981.

\bibitem{gutman:1980}
------, \emph{\BIBforeignlanguage{English}{Controllers for Bilinear Systems}},
  ser. Technical Reports TFRT-7210.\hskip 1em plus 0.5em minus 0.4em\relax
  Department of Automatic Control, Lund Institute of Technology (LTH), 1980.

\bibitem{derese:noldus:1980}
I.~Derese and E.~Noldus, ``Design of linear feedback laws for bilinear
  systems,'' \emph{International Journal of Control}, vol.~31, no.~2, pp.
  219--237, 1980.

\bibitem{luesink:nijmeijer:1989}
R.~Luesink and H.~Nijmeijer, ``On the stabilization of bilinear systems via
  constant feedback,'' \emph{Linear algebra and its applications}, vol. 122,
  pp. 457--474, 1989.

\bibitem{benallou:mellichamp:seborg:1988}
A.~Benallou, D.~A. Mellichamp, and D.~E. Seborg, ``Optimal stabilizing
  controllers for bilinear systems,'' \emph{International Journal of Control},
  vol.~48, no.~4, pp. 1487--1501, 1988.

\bibitem{lin:byrnes:1994}
W.~Lin and C.~I. Byrnes, ``{KYP} lemma, state feedback and dynamic output
  feedback in discrete-time bilinear systems,'' \emph{Systems \& Control
  Letters}, vol.~23, no.~2, pp. 127--136, 1994.

\bibitem{amato:cosentino:fiorillo:merola:2009}
F.~Amato, C.~Cosentino, A.~S. Fiorillo, and A.~Merola, ``Stabilization of
  bilinear systems via linear state-feedback control,'' \emph{IEEE Transactions
  on Circuits and Systems II: Express Briefs}, vol.~56, no.~1, pp. 76--80,
  2009.

\bibitem{khlebnikov:2018}
M.~V. Khlebnikov, ``Quadratic stabilization of discrete-time bilinear
  systems,'' \emph{Automation and Remote Control}, vol.~79, pp. 1222--1239,
  2018.

\bibitem{petersen:1987}
I.~R. Petersen, ``A stabilization algorithm for a class of uncertain linear
  systems,'' \emph{Systems \& Control Letters}, vol.~8, no.~4, pp. 351--357,
  1987.

\bibitem{khlebnikov:shcherbakov:2008}
M.~V. Khlebnikov and P.~S. Shcherbakov, ``Petersen’s lemma on matrix
  uncertainty and its generalizations,'' \emph{Automation and Remote Control},
  vol.~69, no.~11, pp. 1932--1945, 2008.

\bibitem{strasser:berberich:allgower:2023}
R.~Str{\"a}sser, J.~Berberich, and F.~Allg{\"o}wer, ``Robust data-driven
  control for nonlinear systems using the {Koopman} operator,'' in \emph{Proc.
  22nd IFAC World Congress}, 2023, to appear, preprint on arXiv:2304.03519.

\bibitem{huang:jadbabaie:1999}
Y.~Huang and A.~Jadbabaie, ``Nonlinear {$H_\infty$} control: An enhanced
  quasi-{LPV} approach,'' \emph{IFAC Proceedings Volumes}, vol.~32, no.~2, pp.
  2754--2759, 1999.

\bibitem{coutinho:desouza:dasilva:caldeira:prieur:2019}
D.~Coutinho, C.~E. de~Souza, J.~M.~G. da~Silva, A.~F. Caldeira, and C.~Prieur,
  ``Regional stabilization of input-delayed uncertain nonlinear polynomial
  systems,'' \emph{IEEE Transactions on Automatic Control}, vol.~65, no.~5, pp.
  2300--2307, 2019.

\bibitem{scherer:2001}
C.~W. Scherer, ``{LPV} control and full block multipliers,'' \emph{Automatica},
  vol.~37, no.~3, pp. 361--375, 2001.

\bibitem{veenman:scherer:2014}
J.~Veenman and C.~W. Scherer, ``A synthesis framework for robust
  gain-scheduling controllers,'' \emph{Automatica}, vol.~50, no.~11, pp.
  2799--2812, 2014.

\bibitem{lawrence:rugh:1993}
D.~A. Lawrence and W.~J. Rugh, ``Gain scheduling dynamic linear controllers for
  a nonlinear plant,'' \emph{Automatica}, vol.~31, no.~3, pp. 381--390, 1995.

\bibitem{leith:leithead:2000}
D.~J. Leith and W.~E. Leithead, ``Survey of gain-scheduling analysis and
  design,'' \emph{International journal of control}, vol.~73, no.~11, pp.
  1001--1025, 2000.

\bibitem{zhou:doyle:glover:1996}
K.~Zhou, J.~C. Doyle, K.~Glover \emph{et~al.}, \emph{Robust and optimal
  control}.\hskip 1em plus 0.5em minus 0.4em\relax Prentice Hall New Jersey,
  1996, vol.~40.

\bibitem{berberich:scherer:allgower:2022}
J.~Berberich, C.~W. Scherer, and F.~Allg{\"o}wer, ``Combining prior knowledge
  and data for robust controller design,'' \emph{IEEE Transactions on Automatic
  Control}, 2022.

\bibitem{scherer:weiland:2000}
C.~W. Scherer and S.~Weiland, ``Linear matrix inequalities in control,''
  \emph{Lecture Notes, Dutch Institute for Systems and Control, Delft, The
  Netherlands}, vol.~3, no.~2, 2000.

\bibitem{boyd:vandenberghe:2004}
S.~P. Boyd and L.~Vandenberghe, \emph{Convex Optimization}.\hskip 1em plus
  0.5em minus 0.4em\relax Cambridge University Press, 2004.

\bibitem{doyle:glover:khargonekar:francis:1989}
J.~C. Doyle, K.~Glover, P.~P. Khargonekar, and B.~A. Francis, ``State-space
  solutions to standard {$H_2$} and {$H_\infty$} control problems,'' \emph{IEEE
  Transactions on Automatic Control}, vol.~34, no.~8, pp. 831--847, 1989.

\bibitem{packard:zhou:pandey:leonhardson:balas:1992}
A.~Packard, K.~Zhou, P.~Pandey, J.~Leonhardson, and G.~Balas, ``Optimal,
  constant {I/O} similarity scaling for full-information and state-feedback
  control problems,'' \emph{Systems \& Control Letters}, vol.~19, no.~4, pp.
  271--280, 1992.

\bibitem{astolfi:1997}
A.~Astolfi, ``On the relation between state feedback and full information
  regulators in nonlinear singular {$H_\infty$} control,'' \emph{IEEE
  Transactions on Automatic Control}, vol.~42, no.~7, pp. 984--988, 1997.

\bibitem{scherer:2000}
C.~W. Scherer, ``Robust mixed control and linear parameter-varying control with
  full block scalings,'' in \emph{Advances in linear matrix inequality methods
  in control}.\hskip 1em plus 0.5em minus 0.4em\relax SIAM, 2000, pp. 187--207.

\bibitem{lofberg:2004}
J.~{L\"{o}fberg}, ``{YALMIP}: A toolbox for modeling and optimization in
  {MATLAB},'' in \emph{Proc. {IEEE} International Conference on Robotics and
  Automation}, 2004, pp. 284--289.

\bibitem{mosek:2022}
M.~ApS, \emph{The {MOSEK} optimization toolbox for {MATLAB} manual. Version
  9.3.21}, 2022.

\bibitem{vatani:hovd:olaru:2014}
M.~Vatani, M.~Hovd, and S.~Olaru, ``Control design and analysis for discrete
  time bilinear systems using sum of squares methods,'' in \emph{Proc. 53rd
  IEEE Conference on Decision and Control (CDC)}, 2014, pp. 3143--3148.

\bibitem{bisoffi:depersis:tesi:2020a}
A.~Bisoffi, C.~De~Persis, and P.~Tesi, ``Data-based stabilization of unknown
  bilinear systems with guaranteed basin of attraction,'' \emph{Systems \&
  Control Letters}, vol. 145, p. 104788, 2020.

\bibitem{strasser:berberich:allgower:2021}
R.~Str{\"a}sser, J.~Berberich, and F.~Allg{\"o}wer, ``Data-driven control of
  nonlinear systems: Beyond polynomial dynamics,'' in \emph{Proc. 60th IEEE
  Conference on Decision and Control (CDC)}, 2021, pp. 4344--4351.

\bibitem{guo:depersis:tesi:2021}
M.~Guo, C.~De~Persis, and P.~Tesi, ``Data-driven stabilization of nonlinear
  polynomial systems with noisy data,'' \emph{IEEE Transactions on Automatic
  Control}, 2021.

\bibitem{sinha:nandanoori:drgona:vrabie:2022}
S.~Sinha, S.~P. Nandanoori, J.~Drgona, and D.~Vrabie, ``Data-driven
  stabilization of discrete-time control-affine nonlinear systems: A {Koopman}
  operator approach,'' in \emph{Proc. IEEE European Control Conference (ECC)},
  2022, pp. 552--559.

\end{thebibliography}

\appendix
\subsection{Proof of "only if" direction of Proposition~\ref{prop:multiplier-tightness}}\label{sec:appendix-proof-prop-multipliers}
A similar statement is presented in~\cite[Prop.~1]{berberich:scherer:allgower:2022} for repeated \emph{scalar} uncertainties. Our proof follows the same arguments to adapt the result to repeated \emph{vectors}. We structure the proof in three parts. First, we show that every $\Delta\in\mathbf{\Delta}$ is of the form $\Delta = \mathrm{diag}(z_1,...,z_m)$ with $z_1,...,z_m\in\bbR^N$, and second, we prove that $z_i=z_k$ for arbitrary $i,k\in\{1,...,m\}$, i.e., $\Delta=I_m\kron z$. Lastly, we show that $z\in\cZ$.
    
1) 
    Choosing $\Lambda = e_ke_k^\top\succeq 0$ with the $k$-th unit vector $e_k$ and plugging it into~\eqref{eq:Delta-representation-Delta-structured} yields 
    \begin{multline}\label{eq:pf-uncertainty-ek}
        \Delta^\top\left[(e_ke_k^\top)\kron Q_z\right] \Delta 
        + \Delta^\top\left[(e_ke_k^\top)\kron S_z\right] 
        \\
        + \left[(e_ke_k^\top)\kron S_z^\top\right]\Delta
        + \left[(e_ke_k^\top)\kron R_z\right]
        \succeq 0
    \end{multline}
    for $k=1,...,m$.
    First, we observe 
    \begin{align*}
        (e_ke_k^\top)\kron Q_z 
        &= (e_ke_k^\top)\kron (Q_z I_N)
        = (e_k\kron Q_z)(e_k^\top \kron I_N) 
        \\
        &= (e_k\kron I_N) (1\kron Q_z) (e_k^\top \kron I_N) 
        \\
        &= (e_k\kron I_N) Q_z (e_k^\top \kron I_N)
    \end{align*}
    by using $(a\kron b)(c\kron d) = (ac)\kron (bd)$.
    Then, multiplying~\eqref{eq:pf-uncertainty-ek} by $e_i^\top$ and $e_i$ from left and right, respectively, for any $i$ with $i\neq k$, we obtain 
    \begin{equation*}
        0
        \preceq 
        \left[e_i^\top \Delta^\top(e_k\kron I_N)\right] 
            Q_z
        \left[(e_k^\top\kron I_N) \Delta e_i\right],
    \end{equation*}
    where we use $e_i=(e_i\kron 1)$ and $e_i^\top e_k = 0$.
    Since $Q_z\prec 0$, we deduce $(e_k^\top \kron I_N)\Delta e_i = 0$, i.e., the $i$-th column of the $k$-th block row of length $N$ in $\Delta$ is zero. As $i\neq k$ was arbitrary for a fixed $k$, this holds for all columns except for the $k$-th one. Thus, the uncertainty has exactly the claimed structure $\Delta=\diag(z_1,...,z_m)$.

2) 
    Let $i,k\in\{1,...,m\}$ be arbitrary and choose $\Lambda = (e_i-e_k)(e_i-e_k)^\top \succeq 0$. 
    Plugging $\Lambda$ into~\eqref{eq:Delta-representation-Delta-structured} and multiplying the inequality defining the set $\mathbf{\Delta}$ by $(e_i+e_k)^\top$ and $(e_i+e_k)$ from left and right, respectively, leads to
    \begin{equation*}
        \begin{bmatrix}\star\end{bmatrix}^\top
        Q_z
        \left(\left[(e_i-e_k)^\top\kron I_N\right] \Delta (e_i+e_k)\right)
        \succeq 0.
    \end{equation*}
    Using $Q_z\prec 0$, we deduce $\left[(e_i-e_k)^\top\kron I_N\right] \Delta (e_i+e_k)=0$, i.e., $z_i = z_k$.

3) 
    First, we observe that
    \begin{subequations}\label{eq:multiplier-trafo}
        \begin{align}
            \Pi_{\Delta} 
            &= T \left(
                \Lambda \kron \begin{bmatrix} Q_z & S_z \\ S_z^\top & R_z \end{bmatrix}
            \right) T^\top,
            \\
            T 
            &= \begin{bmatrix}
                I_m \kron \begin{bmatrix}
                    I_N & 0_{N\times 1}
                \end{bmatrix} \\
                I_m \kron \begin{bmatrix}
                    0_{1\times N} & 1
                \end{bmatrix}
            \end{bmatrix}.
        \end{align}
    \end{subequations}
    Then, by $(I_m\kron z)\in\mathbf{\Delta}$ we directly conclude 
    \begin{equation*}
        0 \preceq 
        \begin{bmatrix}\star\end{bmatrix}^\top
        \Pi_{\Delta}
        \begin{bmatrix}
            I_m\kron z \\ I
        \end{bmatrix}
        = \Lambda \kron \left(
            \begin{bmatrix}
                z \\ 1
            \end{bmatrix}^\top
            \begin{bmatrix}
                Q_z & S_z \\ S_z^\top & R_z
            \end{bmatrix}
            \begin{bmatrix}
                z \\ 1
            \end{bmatrix}
        \right).
    \end{equation*}
    Since $\Lambda\succeq 0$ is arbitrary, we choose $\Lambda=I_m$ to infer that~\eqref{eq:condition-z-in-Z} holds.
    Hence, we deduce $z\in\cZ$ which concludes the proof.

\end{document}